\documentclass[lettersize,hidelinks,conference]{IEEEtran}

\usepackage{cite}
\usepackage{amsmath,amssymb,amsfonts}
\usepackage{graphicx}
\usepackage{textcomp}
\usepackage[dvipsnames]{xcolor}
\usepackage[acronym,shortcuts]{glossaries}
\usepackage{bm}
\usepackage{caption}
\usepackage{subcaption}
\usepackage{relsize}
\usepackage{soul}
\usepackage{algorithm, algpseudocode}
\usepackage{algcompatible}
\usepackage{outlines}
\usepackage{orcidlink}
\usepackage{lipsum,comment}
\usepackage{mathtools}
\usepackage{tabularx,flushend}
\usepackage{amsthm}
\newtheorem{remark}{Remark}

\theoremstyle{definition}

\newtheorem{proposition}{Proposition}

\def\BibTeX{{\rm B\kern-.05em{\sc i\kern-.025em b}\kern-.08em
T\kern-.1667em\lower.7ex\hbox{E}\kern-.125emX}}

\newcommand{\herm}[0]{^{\mathsf{H}}}

\newcommand{\Real}[1]{\Re\{{#1}\}}
\newcommand{\Imag}[1]{\Im\{{#1}\}}

\newacronym{RPE}{RPE}{radar parameter estimation}
\newacronym{OTFS}{OTFS}{orthogonal time frequency space}
\newacronym{AFDM}{AFDM}{affine frequency division multiplexing}
\newacronym{MIMO}{MIMO}{multiple-input multiple-output}
\newacronym{SISO}{SISO}{single-input single-output}
\newacronym{ISAC}{ISAC}{integrated sensing and communications}
\newacronym{3D}{3D}{three-dimensional}
\newacronym{2D}{2D}{two-dimensional}
\newacronym{1D}{1D}{one-dimensional}
\newacronym{RX}{RX}{receiver}
\newacronym{TX}{TX}{transmitter}
\newacronym{BF}{BF}{beamforming}
\newacronym{mmWave}{mmWave}{millimeter-wave}
\newacronym{SotA}{SotA}{state-of-the-art}
\newacronym{ULA}{ULA}{uniform linear array}
\newacronym{QAM}{QAM}{quadrature amplitude modulation}
\newacronym{ISFFT}{ISFFT}{inverse symplectic finite Fourier transform}
\newacronym{SFFT}{SFFT}{symplectic finite Fourier transform}
\newacronym{AWGN}{AWGN}{additive white Gaussian noise}
\newacronym{OFDM}{OFDM}{orthogonal frequency division multiplexing}
\newacronym{OCDM}{OCDM}{orthogonal chirp division multiplexing}
\newacronym{BS}{BS}{base station}
\newacronym{UE}{UE}{user equipment}
\newacronym{DFT}{DFT}{discrete Fourier transform}
\newacronym{IDFT}{IDFT}{inverse discrete Fourier transform}
\newacronym{TD}{TD}{time-domain}
\newacronym{wlg}{wlg}{without loss of generality}
\newacronym{CP}{CP}{cyclic prefix}
\newacronym{DAFT}{DAFT}{discrete affine Fourier transform}
\newacronym{IDAFT}{IDAFT}{inverse discrete affine Fourier transform}
\newacronym{CPP}{CPP}{\textit{chirp-periodic} prefix}
\newacronym{IDZT}{IDZT}{inverse discrete Zak transform}
\newacronym{DZT}{DZT}{discrete Zak transform}
\newacronym{ICI}{ICI}{inter-carrier interference}
\newacronym{BER}{BER}{bit error rate}
\newacronym{DoF}{DoF}{degrees-of-freedom}
\newacronym{FD}{FD}{full-duplex}
\newacronym{SIMO}{SIMO}{single-input multiple-output}
\newacronym{MISO}{MISO}{multiple-input single-output}
\newacronym{AoD}{AoD}{angle-of-departure}
\newacronym{AoA}{AoA}{angle-of-arrival}
\newacronym{RF}{RF}{radio frequency}
\newacronym{SIM}{SIM}{stacked intelligent metasurfaces}
\newacronym{FIM}{FIM}{flexible intelligent metasurface}
\newacronym{FPGA}{FPGA}{field programmable gate array}
\newacronym{UPA}{UPA}{uniform planar array}
\newacronym{CC}{CC}{communication-centric}
\newacronym{I/O}{I/O}{input-output}
\newacronym{iid}{i.i.d.}{independent and identically distributed}
\newacronym{IoT}{IoT}{internet of things}
\newacronym{V2X}{V2X}{vehicle-to-everything}
\newacronym{NTN}{NTN}{non-terrestrial network}
\newacronym{LEO}{LEO}{low-earth orbit}
\newacronym{THz}{THz}{terahertz}
\newacronym{EM}{EM}{electromagnetic}
\newacronym{STAR-RIS}{STAR-RIS}{simultaneously transmitting and reflecting reconfigurable intelligent surface}
\newacronym{DoA}{DoA}{direction-of-arrival}
\newacronym{DD}{DD}{doubly-dispersive}
\newacronym{ODDM}{ODDM}{orthogonal delay-Doppler division multiplexing}
\newacronym{LoS}{LoS}{line-of-sight}
\newacronym{NLoS}{NLoS}{non-line-of-sight}
\newacronym{6G}{6G}{sixth generation}
\newacronym{MPDD}{MPDD}{metasurfaces-parameterized DD}
\newacronym{GaBP}{GaBP}{Gaussian Belief Propagation}
\newacronym{MSE}{MSE}{mean-squared-error}
\newacronym{sIC}{soft IC}{soft interference cancellation}
\newacronym{soft RG}{soft RG}{soft replica generation}
\newacronym{BG}{BG}{belief generation}
\newacronym{SGA}{SGA}{scalar Gaussian approximation}
\newacronym{CLT}{CLT}{central limit theorem}
\newacronym{PDF}{PDF}{probability density function}
\newacronym{QPSK}{QPSK}{quadrature phase-shift keying}
\newacronym{LMMSE}{LMMSE}{linear minimum mean square error}
\newacronym{SNR}{SNR}{signal-to-noise ratio}
\newacronym{QoS}{QoS}{quality of service}
\newacronym{CoV}{CoV}{calculus of variations}
\newacronym{CAPA}{CAPA}{continuous aperture array}
\newacronym{FCAPA}{FCAPA}{flexible continuous aperture array}
\newacronym{GL}{GL}{Gauss-Legendre}
\newacronym{DDC MIMO}{DDC MIMO}{DD continuous MIMO}
\newacronym{B5G}{B5G}{beyond fifth generation}
\newacronym{VR}{VR}{virtual reality}
\newacronym{XR}{XR}{extended reality}
\newacronym{ITN}{ITN}{intelligent traffic networks}
\newacronym{SAGIN}{SAGIN}{space-air-ground integrated network}
\newacronym{UAV}{UAV}{unmanned aerial vehicle}
\newacronym{MUSIC}{MUSIC}{Multiple Signal Classification}
\newacronym{ICC}{ICC}{integrated communication and computing}
\newacronym{SINR}{SINR}{signal-to-interference-plus-noise ratio}
\newacronym{WSR}{WSR}{weighted sum rate}
\newacronym{ARPU}{ARPU}{average rate per user}
\newacronym{BCD}{BCD}{block coordinate descent}
\newacronym{PDE}{PDE}{partial differential equation}
\newacronym{EL}{EL}{Euler-Lagrange}
\newacronym{TCA}{TCA}{tightly coupled array}
\newacronym{ELAA}{ELAA}{extremely large-aperture arrays}
\newacronym{LIS}{LIS}{large intelligent surface}
\newacronym{CSI}{CSI}{channel state information}
\newacronym{RIS}{RIS}{reconfigurable intelligent surface}
\newacronym{SEP}{SEP}{symbol error probability}
\newacronym{ADC}{ADC}{analog-to-digital converter}
\newacronym{MRC}{MRC}{maximum ratio combining}
\newacronym{MMA}{MMA}{moment matching approximation}

\begin{document}

\title{1-bit Quantized Continuous Aperture Arrays}

\author{\IEEEauthorblockN{Kuranage Roche Rayan Ranasinghe$^*$, Getuar Rexhepi$^*$, Zhaolin Wang$^\dag$ and Giuseppe Thadeu Freitas de Abreu$^*$}
\IEEEauthorblockA{$^*$\textit{School of Computer Science and Engineering, Constructor University, 28759 Bremen, Germany} \\
$^\dag$\textit{Department of Electrical and Computer Engineering, The University of Hong Kong, Hong Kong} \\
Emails: [kranasinghe, grexhepi, gabreu]@constructor.university, zhaolin.wang@hku.hk}
\vspace{-4ex}
}



\maketitle

\begin{abstract}
\Acp{CAPA} have emerged as a promising physical-layer paradigm for \ac{6G} systems, offering spatial degrees of freedom beyond those of conventional discrete antenna arrays. 
This paper investigates the interaction between the \ac{CAPA} receive architecture and low-cost 1-bit \acp{ADC}, which impose a severe nonlinear distortion penalty in conventional discrete systems. 
For Rayleigh fading, we derive a \ac{MMA}-based closed-form \ac{SEP} approximation based on Gamma moment-matching of the spatial eigenvalue distribution, and show that \acp{CAPA} incur a diversity-order penalty governed by Jensen's inequality on the mode eigenvalues. 
For \ac{LoS} propagation, we prove that \ac{CAPA} achieves \emph{exactly} the unquantized \ac{AWGN} performance bound under perfect spatial and phase alignment, completely eliminating the 1-bit penalty that forces discrete systems to double their antenna count. 
Monte Carlo simulations under Rayleigh, Rician, and \ac{LoS} conditions validate all analytical results.
\end{abstract}


\glsresetall


\section{Introduction}

\IEEEPARstart{M}{ulti-antenna} technology has been central to every generation of wireless systems, from spatial multiplexing in fourth-generation long-term evolution to the massive array deployments of fifth-generation new radio~\cite{larsson2014massive, 10144733}. 
Looking toward \ac{6G}, \acp{CAPA} have been identified as a transformative architecture: by treating the receive surface as a continuous electromagnetic structure rather than a collection of spatially isolated elements, \acp{CAPA} exploit the full spatial bandwidth of the propagation channel and approach the information-theoretic limits imposed by the electromagnetic degrees of freedom (DoF)~\cite{dang2020should,HanTWC2023,OuyangTWC2025}. 
Unlike discrete arrays, a \ac{CAPA} integrates the incident field over the entire physical aperture before any digital processing occurs, a property that motivates the present investigation.
 
A parallel hardware trend in massive \ac{MIMO} has driven interest in receivers equipped with low-resolution \acp{ADC}, and particularly 1-bit \acp{ADC}, which reduce power consumption and circuit complexity by several orders of magnitude relative to their high-resolution counterparts~\cite{JacobssonTWC2017}. 
However, 1-bit quantization imposes a severe nonlinear distortion that fundamentally limits the achievable diversity order. 
Recent work has rigorously characterized this penalty for discrete \ac{SIMO} systems under \ac{QPSK} signaling, establishing that achieving the same high-\ac{SNR} \ac{SEP} scaling as an unquantized receiver requires approximately doubling the number of receive antennas~\cite{ravinath2026sepjournal,MoTSP2015}. 
This ``twice the antennas'' rule represents a significant hardware cost and motivates the search for receiver architectures that can circumvent it.
 
Whether the continuous spatial integration inherent to a \ac{CAPA} can mitigate the 1-bit quantization penalty remains an open question. 
The answer is non-trivial: the analog beamforming stage of a \ac{CAPA} operates \emph{before} the \acp{ADC}, meaning the field is spatially combined in the continuous domain prior to quantization. 
This pre-quantization integration fundamentally differs from the per-antenna quantization of discrete arrays, and its interaction with the nonlinear 1-bit distortion has not previously been studied.
 
This paper makes the following contributions. \emph{(i)} Under Rayleigh fading, we derive a \ac{MMA}-based \ac{SEP} approximation that quantifies the eigenvalue-dispersion penalty of the \ac{CAPA} architecture relative to an idealized \ac{iid} \ac{SIMO} baseline via a Gamma moment-matching argument. \emph{(ii)} Under pure \ac{LoS} propagation, we prove analytically (Proposition~1) that a 1-bit \ac{CAPA} with perfect spatial and phase alignment achieves exactly the unquantized \ac{AWGN} \ac{SEP} bound, with the $2/\pi$ ADC penalty vanishing identically. \emph{(iii)} We provide Monte Carlo validation across Rayleigh, Rician, and pure \ac{LoS} regimes, including a spatial misalignment experiment that characterizes graceful degradation from the ideal \ac{LoS} bound.

\section{System Model}
\label{sec:Sys_Model_single_user}

This section establishes the system model for a single-user \ac{CAPA}-based uplink scenario. 

\subsection{SIMO System Model}

Consider a traditional received signal vector $\mathbf{y} \in \mathbb{C}^{N \times 1}$ at an uplink \ac{SIMO} system given by \cite{ravinath2026sepjournal}
\begin{equation}\label{eq:simo_model}
    \mathbf{y} = \sqrt{P} \mathbf{h} s + \mathbf{n},
\end{equation}
where $P$ denotes the average transmit power, $\mathbf{h} \in \mathbb{C}^{N \times 1}$ defines the discrete spatial channel response vector (assuming a rich scattering environment with \ac{iid} Rayleigh fading, the channel coefficients are modeled as circularly symmetric complex Gaussian random variables $\mathbf{h} \sim \mathcal{CN}(\mathbf{0}, \mathbf{I}_N)$), $s \in \mathbb{C}$ is the transmitted data symbol drawn from a uniform discrete constellation (e.g., \ac{QPSK} or $M$-PSK) statistically normalized to possess unit average energy such that $\mathbb{E}[|s|^2] = 1$, and $\mathbf{n} \sim \mathcal{CN}(\mathbf{0}, \sigma_n^2 \mathbf{I}_N)$ is the additive white Gaussian noise vector.

Let the \ac{ADC} quantizer operation be defined as \cite{MoTSP2015}
\vspace{-1ex}
\begin{equation}
\label{eq:quantizer_def}
\vspace{-1ex}
    \mathrm{Q}(\cdot) \triangleq \frac{1}{\sqrt{2}} \Big( \text{sgn}\big( \Real{\cdot} \big) + j \text{sgn}\big( \Imag{\cdot} \big) \Big).    
\end{equation}

Then, the quantized vector $\tilde{\mathbf{y}}$ at the output of the \acp{ADC} is given by
\vspace{-1ex}
\begin{equation}
    \tilde{\mathbf{y}} = \mathrm{Q}(\mathbf{y}) \in \mathcal{S}.
\end{equation}

Using \ac{MRC}-based \ac{QPSK} symbol detection with fully digital beamforming\footnote{In principle, one can also incorporate analog beamforming to the above procedure.
However, the fully digital framework yields a best-case bound, and is therefore used persistently throughout.} yields
\vspace{-1ex}
\begin{equation}
    \hat{s}_\mathrm{MRC} = \mathrm{Q}(\mathbf{h}\herm \tilde{\mathbf{y}}) \in \mathcal{S}.
\end{equation}

It is important to note that the exact analytical \ac{SEP} bounds derived in \cite{ravinath2026sepjournal} specifically characterize the \ac{SIMO} architecture under Rayleigh fading with \ac{QPSK} signaling. 
Consequently, our theoretical \ac{SIMO} baseline comparisons are evaluated under these strict conditions. 
The derivation of closed-form 1-bit bounds for generalized continuous Rician environments remains an open problem and is reserved for future work.

\subsection{CAPA System Model}

According to \cite{OuyangTWC2025,ZhaoTWC2025}, the receive signal in a single-user uplink \ac{CAPA} system at a point $\mathbf{r}$ within the aperture $\mathcal{R}$ spanning an area $|\mathcal{R}| \triangleq A = L_x L_z$ in the $xz$-plane can be expressed as
\begin{equation}
\label{eq:single_user_model}
    y(\mathbf{r}) = \sqrt{P} g(\mathbf{r}) s + n(\mathbf{r}), \quad \mathbf{r} \in \mathcal{R}, 
\end{equation}
where $P$ denotes the average transmit power, $g(\mathbf{r})$ is the continuous spatial channel response, $s \in \mathcal{S}$ is the \ac{QPSK} symbol with $\mathbb{E}[|s|^2]=1$, and $\mathbb{E}[n(\mathbf{r}) n^*(\mathbf{r'})] = \sigma^2_n \delta(\mathbf{r}- \mathbf{r'})$. 

\begin{remark}[Channel Energy Normalization]
\label{Remark:chan_energy_norm}
    To ensure a fair mathematical comparison against the discrete \ac{SIMO} baseline -- where each antenna is subject to unit average channel gain -- the continuous spatial channel response $g(\mathbf{r})$ is strictly normalized. Specifically, the field is scaled such that the expected electromagnetic energy captured by the aperture across the $M$ active spatial degrees of freedom satisfies $\sum_{m=1}^{M} \mathbb{E}\big[|h_m|^2\big] = M$. For a general Rician channel, this normalization imposes
    \begin{equation}
        \frac{K}{K+1}\sum_{m=1}^{M} |h_{m,\mathrm{LoS}}|^2 + \frac{1}{K+1}\sum_{m=1}^{M} \mathbb{E}\big[|h_{m,\mathrm{NLoS}}|^2\big] = M.
    \end{equation}
    
    Consequently, in the limiting case of a pure \ac{LoS} environment ($K \to \infty$) under perfect spatial alignment, the \ac{NLoS} components vanish, and this normalization guarantees that the entirety of the $M$-unit energy budget collapses into the dominant deterministic mode, yielding exactly $|h_{1,\mathrm{LoS}}|^2 = M$.
\end{remark}

\subsubsection{Channel Model}

The generalized Rician channel response is defined as
\begin{equation}
    g(\mathbf{r}) = \sqrt{\frac{K}{K+1}} g_{\mathrm{LoS}}(\mathbf{r}) + \sqrt{\frac{1}{K+1}} g_{\mathrm{NLoS}}(\mathbf{r}),
\end{equation}
where $K \ge 0$ is the Rician factor. 

The deterministic free-space propagation $g_{\mathrm{LoS}}(\mathbf{r})$ from a given transmit point $\mathbf{p}$ is modeled by the scalar Green's function given as
\begin{equation}
    g_{\mathrm{LoS}}(\mathbf{r}) \triangleq \frac{1}{4\pi \| \mathbf{r} - \mathbf{p} \|} e^{-j \frac{2\pi}{\lambda} \| \mathbf{r} - \mathbf{p} \|}.
\end{equation}

Because the continuous aperture $\mathcal{R}$ has finite physical dimensions, projecting $g_{\mathrm{LoS}}(\mathbf{r})$ onto a fixed orthogonal spatial Fourier basis generally results in spectral leakage. 
Unless the user's angle of arrival precisely aligns with a spatial frequency bin, the \ac{LoS} energy forms a deterministic vector $\mathbf{h}_{\mathrm{LoS}} \in \mathbb{C}^{M \times 1}$ characterized by a spatial $\mathrm{sinc}$ pattern across the $M$ active modes \cite{9906802}. 
However, under optimal analog beamforming, the continuous combining functions can be perfectly matched to the user's angle of arrival. 
In this ideal spatially aligned scenario, spectral leakage vanishes, confining the aperture energy to a single dominant spatial mode.

Following electromagnetic scattering theory for a single-antenna transmitter, the scattering environment maps the spherical wave emitted from the point source into a superposition of plane waves arriving at the receive aperture. The \ac{NLoS} spatial channel $g_{\mathrm{NLoS}}(\mathbf{r})$ is represented as a continuous spatial Fourier integral over the arriving wave vectors $\mathbf{k}$, given by
\begin{equation}
    g_{\mathrm{NLoS}}(\mathbf{r}) = \frac{1}{2\pi} \int_{\mathcal{D}_w(\mathbf{k})} \!\!\!\!\! H_a(\mathbf{k}) e^{-j \mathbf{k}^T \mathbf{r}} dk_x dk_z.
\end{equation}

Here, $H_a(\mathbf{k})$ is the angular wavenumber-domain response, stochastically modeled as
\begin{equation}
    H_a(\mathbf{k}) = \frac{A_a(\mathbf{k}) W_a(\mathbf{k})}{\sqrt{\gamma(k_x, k_z)}},
\end{equation}
where $A_a(\mathbf{k})$ is a deterministic spectral factor mapping the physical cluster geometries. The term $\gamma(k_x, k_z) = \sqrt{(2\pi/\lambda)^2 - k_x^2 - k_z^2}$ is the longitudinal wavenumber component, appearing as a geometric Jacobian factor in the transformation from angular solid-angle coordinates to transverse wavenumber coordinates. 

Because the receiver processes only a truncated subspace of $M$ active modes out of the total $D \approx |\mathcal{R}|/(\lambda/2)^2$ available electromagnetic degrees of freedom ($M < D$), a portion of the total incident field energy inherently falls outside the spatial bandwidth of the combiner.
Therefore, as a modeling convention to satisfy the strict energy budget established in Remark \ref{Remark:chan_energy_norm}, the spectral factor $A_a(\mathbf{k})$ is structurally scaled such that the expected energy captured specifically by the $M$ active combining functions satisfies $\sum_{m=1}^M \mathbb{E}[|h_{m,\mathrm{NLoS}}|^2] = M$. 
The term $W_a(\mathbf{k}) \sim \mathcal{CN}(0,1)$ is a complex-Gaussian random field satisfying the uncorrelated scattering condition $\mathbb{E}[W_a(\mathbf{k}) W_a^*(\mathbf{k}')] = \delta(\mathbf{k} - \mathbf{k}')$.

\subsubsection{Analog Beamforming and Detection}

The continuous field is first projected onto the $M$ orthonormal modes, followed by independent scalar 1-bit quantization per mode.
The receiver employs $M$ orthonormal analog beamformers $\{w_m(\mathbf{r}) \in \mathbb{C}\}^M_{m=1}$ to match the $M$ orthogonal spatial modes, constrained by the aforementioned electromagnetic degrees of freedom ($M < D$).
The projected received scalars can then be expressed as
\vspace{-1ex}
\begin{equation}
    y_m = \int_\mathcal{R} w_m^*(\mathbf{r}) y(\mathbf{r}) \mathrm{d} \mathbf{r} = \sqrt{P} h_m s + \eta_m,
\end{equation}
where $\eta_m \sim \mathcal{CN}(0,\sigma^2_n)$, with the beamformers satisfying $\int w_m^*(\mathbf{r}) w_n(\mathbf{r}) d\mathbf{r} = \delta_{mn}$.

The \ac{ADC} output vector is $\tilde{\mathbf{y}} = \mathrm{Q}(\mathbf{y}) \in \mathcal{S}^M$, and the final symbol estimate is obtained via digital \ac{MRC} combining as $\hat{s} = \mathrm{Q}(\mathbf{h}\herm \tilde{\mathbf{y}})$.

\begin{remark}[Spatial Decorrelation and Mode Variance]
    While the uncorrelated scattering model holds in the angular wavenumber domain, the finite continuous aperture $\mathcal{R}$ imposes a fundamental spatial correlation on the received \ac{NLoS} field. Projecting this correlated field onto the orthonormal analog basis yields decorrelated, but \emph{non-identically distributed}, mode gains $h_{m,\mathrm{NLoS}}$ \cite{PizzoJSAC2020}. Consequently, the effective discrete \ac{NLoS} channel vector follows $\mathbf{h}_{\mathrm{NLoS}} \sim \mathcal{CN}(\mathbf{0}, \mathbf{\Lambda})$, where $\mathbf{\Lambda} = \mathrm{diag}(\lambda_1, \dots, \lambda_M)$ contains the unequal eigenvalues of the spatial correlation operator. This diverges fundamentally from the \ac{iid} assumptions of conventional discrete \ac{SIMO} arrays.
\end{remark}

\section{SEP Analysis}
\label{sec:SEP_Analysis}

\subsection{1-bit Quantized SIMO Baseline}

For a conventional uplink \ac{SIMO} system employing $N$ \ac{iid} receive antennas with \ac{MRC} combining, the exact average \ac{SEP} under \ac{QPSK} signaling and Rayleigh fading is rigorously characterized in \cite{ravinath2026sepjournal} as
\begin{equation}
    P_{\mathrm{SE}}^{\mathrm{MRC}} = 2\,\mathbb{E}\!\left\{ Q\!\left( \sqrt{\frac{\rho}{N}U} \right) \right\} - \left( \mathbb{E}\!\left\{ Q\!\left( \sqrt{\frac{\rho}{N}U} \right) \right\} \right)^{2},
\end{equation}
where $Q(\cdot)$ is the Gaussian tail probability function, $\rho \triangleq P/\sigma_n^2$ represents the average received \ac{SNR} per branch, and $U = (\sum_{i=1}^{N} Z_i)^2$, with $\{Z_i\}_{i=1}^{N}$ being \ac{iid} standard half-normal random variables. 
This establishes that achieving unquantized \ac{SEP} scaling requires approximately twice the number of receive antennas.

\subsection{1-bit Quantized CAPA Analysis}

Because the effective \ac{CAPA} channel $\mathbf{h}$ possesses unequal mode variances $\mathbf{\Lambda}$, the continuous architecture is not mathematically isomorphic to the discrete \ac{SIMO} system. 

\subsubsection{MMA-based Bounding for Rayleigh Channels}

It is crucial to note that under pure Rayleigh fading, the \ac{SEP} of the \ac{CAPA} architecture will inherently deviate from the exact \ac{SIMO} bounds. 
Because the continuous spatial correlation restricts the available degrees of freedom, the \ac{iid} \ac{SIMO} bound serves as a strict theoretical upper limit on spatial diversity against which the practical \ac{CAPA} eigenvalue penalty can be quantified.

\textbf{Moment Matching Approximation:}
Under Rayleigh fading, the average \ac{SEP} of the 1-bit quantized \ac{CAPA} architecture can be approximated by modeling the decision statistic as a generalized Gamma variable, yielding
\begin{equation}
\label{eq:CAPA_exact_SEP}
    P_{\mathrm{SE}}^{\mathrm{CAPA}} \approx 2 \mathcal{I}_{\mathbf{\Lambda}} - \mathcal{I}_{\mathbf{\Lambda}}^{2},
\end{equation}
where $\mathcal{I}_{\mathbf{\Lambda}}$ is the expected Gaussian tail probability evaluated over a Gamma distribution defined as
\begin{equation}
\label{eq:CAPA_Integral}
    \mathcal{I}_{\mathbf{\Lambda}} = \int_{0}^{\infty} Q\!\left( \sqrt{\frac{\rho}{M}} w \right) \frac{w^{k-1} e^{-w/\theta}}{\Gamma(k) \theta^k} \mathrm{d}w.
\end{equation}

The shape parameter $k$ and scale parameter $\theta$ are determined by the \ac{NLoS} spatial eigenvalues $\mathbf{\Lambda} = \mathrm{diag}(\lambda_1, \dots, \lambda_M)$ as
\begin{align}
    \label{eq:gamma_k}
    k &= \frac{\frac{2}{\pi} \left( \sum_{m=1}^{M} \sqrt{\lambda_m} \right)^2}{\left( 1 - \frac{2}{\pi} \right) \sum_{m=1}^{M} \lambda_m}, \quad \theta = \frac{\left( 1 - \frac{2}{\pi} \right) \sum_{m=1}^{M} \lambda_m}{\sqrt{\frac{2}{\pi}} \sum_{m=1}^{M} \sqrt{\lambda_m}}.
\end{align}

\emph{Justification:} For \ac{QPSK} signaling ($s = s_R + j s_I$), the unquantized signal at the $m$-th mode is $y_m = \sqrt{P} (\Re\{h_m\}s_R - \Im\{h_m\}s_I) + j\sqrt{P}(\Re\{h_m\}s_I + \Im\{h_m\}s_R) + \eta_m$. 
After 1-bit quantization, the \ac{MRC} decision statistic for the real dimension is given by $z_R = \Re\{\mathbf{h}\herm \mathrm{Q}(\mathbf{y})\}$, which expands to
\begin{equation}
    z_R = \frac{1}{\sqrt{2}} \sum_{m=1}^M \Big[ \Re\{h_m\} \mathrm{sgn}(\Re\{y_m\}) + \Im\{h_m\} \mathrm{sgn}(\Im\{y_m\}) \Big].
\end{equation}

The 1-bit \ac{MRC} combiner tightly cross-couples the real and imaginary channel components, precluding an exact closed-form probability density function. 
Rather than analytically untangling this nonlinear phase coupling, we conjecture, by analogy with the \ac{iid} \ac{SIMO} derivation in \cite{ravinath2026sepjournal}, that the dominant contribution to the error probability can be captured by a generalized effective variable. 
Specifically, we assume the combining variable generalizes the \ac{SIMO} statistic to the case of unequal branch variances, taking the form $W_{\mathbf{\Lambda}} \triangleq \sum_{m=1}^{M} \sqrt{\lambda_m} Z_m$, where $\{Z_m\}$ are \ac{iid} standard half-normal variables.

The effective variable $W_{\mathbf{\Lambda}}$ is structurally normalized to match the convention in \cite{ravinath2026sepjournal}. 
This ensures that the corresponding $Q$-function argument recovers exactly $\sqrt{\rho/N} \sum_i Z_i$ in the \ac{iid} limit ($M=N$, $\lambda_m=1$). 

Since $W_{\mathbf{\Lambda}}$ isolates the fundamental diversity degradation caused by the spatial eigenvalues, we leverage a well-known moment-matching approximation \cite{AnnamalaiTCOM2000,PeppasTVT2018,MaTIT2025}. 
Matching the central moments of $W_{\mathbf{\Lambda}}$ to a Gamma distribution $W_{\mathbf{\Lambda}} \sim \Gamma(k, \theta)$ yields \eqref{eq:gamma_k}. 
Substituting this into the expectation $\mathbb{E}\{Q(\sqrt{\rho/M} W_{\mathbf{\Lambda}})\}$ produces the integral $\mathcal{I}_{\mathbf{\Lambda}}$. 
Crucially, in the fully uncorrelated limit, this approach mathematically recovers the corresponding moment-matched Gamma approximation of the \ac{iid} \ac{SIMO} combining variable, ensuring analytical consistency at the boundary.

\subsubsection{Analytical Expression for Line-of-Sight}

We now present an analytical bound for 1-bit quantized \acp{CAPA} under \ac{LoS} conditions.

\begin{proposition}
\label{prop:Exact_LoS_Bounds}
In a pure \ac{LoS} environment ($K \to \infty$), assuming uniform plane-wave incidence at broadside -- where all antenna elements experience identical deterministic channel gains and zero relative phase shifts -- the exact \ac{SEP} of the 1-bit quantized discrete \ac{SIMO} architecture is
\begin{equation}
\label{eq:SEP_SIMO_LoS}
    \mathrm{SEP}_{\mathrm{SIMO, LoS}} = 2P_{e,R} - P_{e,R}^2,
\end{equation}
where $P_{e,R} = \sum_{\ell=\lfloor\frac{N}{2}\rfloor + 1}^N \binom{N}{\ell} p_0^\ell (1-p_0)^{N-\ell} + \frac{1}{2} \mathbb{I}_{\{N \text{ is even}\}} \binom{N}{N/2} p_0^{N/2} (1-p_0)^{N/2}$, with $p_0 = Q(\sqrt{\rho})$ and $\mathbb{I}_{\{\cdot\}}$ acting as the indicator function.

Conversely, under perfect spatial and phase alignment, the exact \ac{SEP} of the 1-bit quantized continuous \ac{CAPA} architecture achieves exactly the unquantized \ac{AWGN} bound, given by
\begin{equation}
\label{eq:SEP_CAPA_LoS}
    \mathrm{SEP}_{\mathrm{CAPA, LoS}} = 2p_{e,\mathrm{CAPA}} - p_{e,\mathrm{CAPA}}^2,
\end{equation}
where $p_{e,\mathrm{CAPA}} = Q( \sqrt{M \rho} )$.
\end{proposition}

\begin{IEEEproof}
For the \ac{SIMO} architecture, under uniform broadside incidence, the channel coefficients after phase compensation become a common positive scalar $h_i = \alpha > 0$ for all $i$. 
Consistent with the unit-variance \ac{iid} Rayleigh baseline, the deterministic \ac{LoS} path ensures unit gain per antenna ($\alpha = 1$). 
Under \ac{QPSK} signaling ($|s_R| = 1/\sqrt{2}$), the single-branch error probability evaluates exactly to $p_0 = Q(\sqrt{P \alpha^2 (1/2) / (\sigma_n^2/2)}) = Q(\sqrt{\rho})$. 
The digital \ac{MRC} combiner thus reduces to an unweighted sum $z_R = \frac{1}{\sqrt{2}} \sum_{i=1}^N \mathrm{sgn}(\Re\{y_i\})$.
As is standard in classical hard-decision diversity combining, this functionally operates as a true spatial majority voter over $N$ independent Bernoulli trials. 
An error occurs if more than half of the antennas yield an incorrect bit ($\ell > N/2$), yielding the Binomial sum in \eqref{eq:SEP_SIMO_LoS}.

For the \ac{CAPA} architecture, assuming perfect spatial alignment, the analog basis projects the entire aperture energy into a single dominant mode. 
Due to the $K \to \infty$ normalization constraint, the entirety of the aperture's effective gain concentrates into this mode ($|h_1|^2 = M$). 
To realize optimal detection, the analog beamformer must apply a conjugate phase shift to align the complex phase of this dominant mode with the baseband axes prior to quantization. 
While 1-bit quantization typically incurs an asymptotic $2/\pi$ power penalty, this is eliminated under ideal phase-aligned single-mode reception.
Because the 1-bit \ac{ADC} thresholds perfectly coincide with the unrotated \ac{QPSK} decision boundaries, no spatial or amplitude information is lost. 
The digital decision relies entirely on the sign of the $M$-amplified mode, reducing the error probability directly to the unquantized Gaussian tail $Q(\sqrt{M\rho})$.
\end{IEEEproof}

\subsection{Architectural Implications of the Analytical Results}

The analytical results produced via the \ac{MMA} and Proposition \ref{prop:Exact_LoS_Bounds} mathematically formalize the operational trade-offs of a \ac{CAPA} relative to a discrete \ac{SIMO} system \cite{ravinath2026sepjournal}. 
Specifically, the preceding analysis highlights two primary architectural implications:
\begin{itemize}
    \item \textbf{The Rayleigh Eigenvalue Penalty:} Under isotropic scattering, the finite continuous aperture imposes spatial correlation, forcing the \ac{NLoS} energy into a subset of dominant eigenvalues $\mathbf{\Lambda}$. As quantified by the shape parameter $k$ in the \ac{MMA} procedure, a \ac{CAPA} of $M$ modes performs equivalently to an idealized \ac{iid} \ac{SIMO} system experiencing a reduced effective diversity order. Specifically, an \ac{iid} array of $M$ discrete antennas yields a baseline shape parameter $k_{\mathrm{iid}} = \frac{2/\pi}{1-2/\pi}M \approx 1.75M$. For \ac{CAPA}, this parameter degrades to $k < k_{\mathrm{iid}}$, meaning \ac{CAPA} requires larger physical apertures to overcome 1-bit distortion. This degradation is a direct consequence of Jensen's inequality; because the square root is a strictly concave function, any non-uniform eigenvalue distribution subject to the energy constraint $\sum_{m=1}^M \lambda_m = M$ mathematically guarantees $\sum_{m=1}^M \sqrt{\lambda_m} < M$, strictly yielding $k < k_{\mathrm{iid}}$.
    
    \item \textbf{LoS Phase-Aligned Array Gain:} Conversely, in deterministic \ac{LoS} environments, \ac{CAPA} circumvents the discrete ``twice the antennas'' scaling rule entirely. By utilizing analog beamforming to concentrate energy and align the phase of a single spatial mode prior to quantization, \ac{CAPA} perfectly aligns the signal with the \ac{ADC} thresholds, fundamentally eliminating the 1-bit nonlinear penalty and achieving optimal unquantized performance.
\end{itemize}

\section{Performance Analysis}
\label{sec:results}

All results use $N = M = 8$, $f = 2.4$~GHz, and $A = 0.25$~m$^2$ unless stated otherwise. 
The \ac{CAPA} \ac{NLoS} channel follows the isotropic sinc correlation model of~\cite{PizzoJSAC2020}, and the \ac{LoS} component is projected via the scalar Green's function at broadside. 
Each point is averaged over $10^7$ Monte Carlo trials.

\subsection{Rayleigh Fading}

Fig.~\ref{fig:1-bit_Rayleigh} compares the nominal \ac{SotA} \cite{OuyangTWC2025} ($A = 0.25$~m$^2$) and minimum ($A = 0.03125$~m$^2$) aperture cases. 
For the large aperture, only $M/D = 12.5\%$ of available modes are activated, yielding near-uniform eigenvalues ($\lambda_m \in [0.955, 1.029]$) and negligible Jensen penalty; the quantized \ac{CAPA} and \ac{SIMO} curves are therefore nearly coincident, serving as a consistency check. 
For the smaller aperture, $M$ already reaches the practical \acp{DoF} limit (i.e., \(100\%\) of the available modes), indicating that no performance degradation occurs with reduced aperture size as long as the \acp{DoF} are fully utilized.
In both cases, the $W_\Lambda$ conjecture tracks the \ac{CAPA} simulation accurately. 
The closed-form Gamma integral of via \ac{MMA} is reliable at low-to-moderate \ac{SNR} but becomes over-optimistic above approximately $2$~dB, where tail mismatch in the Gamma approximation dominates.

\subsection{Rician Fading}

Fig.~\ref{fig:1-bit_Rice8dB} shows the intermediate regime at $K = 2,8$~dB. 
No closed-form bounds exist for this case. 
At the larger $8$~dB, the 1-bit \ac{CAPA} curve sits measurably closer to the unquantized limit than under Rayleigh fading, as the strong \ac{LoS} component ($86\%$ of received power) begins concentrating energy into the dominant spatial mode and partially recovering the phase-alignment benefit of Proposition~\ref{prop:Exact_LoS_Bounds}. 
The quantization gap for the \ac{SIMO} system similarly narrows due to channel hardening. 
These observations confirm that the Rician regime interpolates smoothly between the two analytically characterized extremes.

\subsection{Pure Line-of-Sight}

The first subfigure in Fig.~\ref{fig:1-bit_LoS_A1} validates Proposition~\ref{prop:Exact_LoS_Bounds} under perfect spatial and phase alignment. 
The \ac{SIMO} Binomial bound and \ac{CAPA} \ac{AWGN} bound both match their respective simulations to within Monte Carlo precision. 

\begin{figure}[H]
    \centering
    \begin{subfigure}{0.95\columnwidth}
        \centering
        \includegraphics[width=\columnwidth]{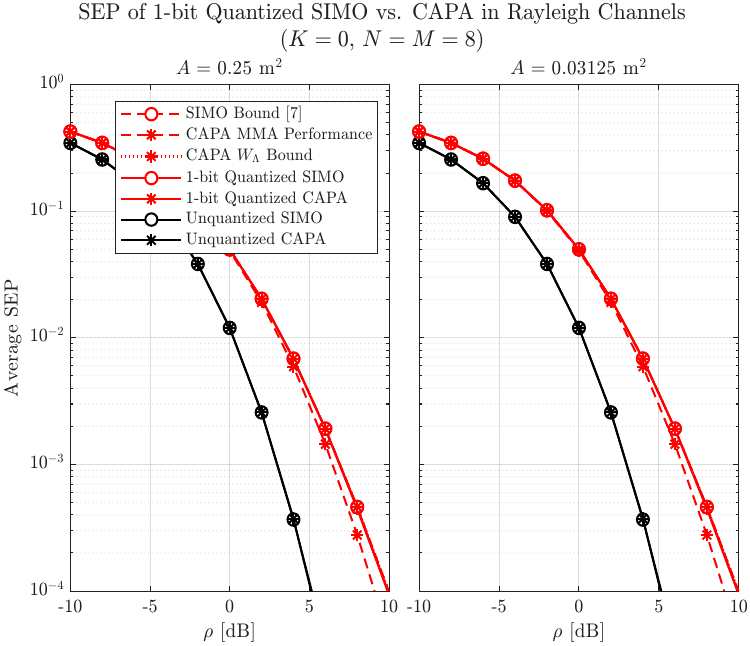}
        \vspace{-4ex}
        \caption{Rayleigh fading.}
        \label{fig:1-bit_Rayleigh}
    \end{subfigure}
    \\
    \begin{subfigure}{0.95\columnwidth}
        \centering
        \includegraphics[width=\columnwidth]{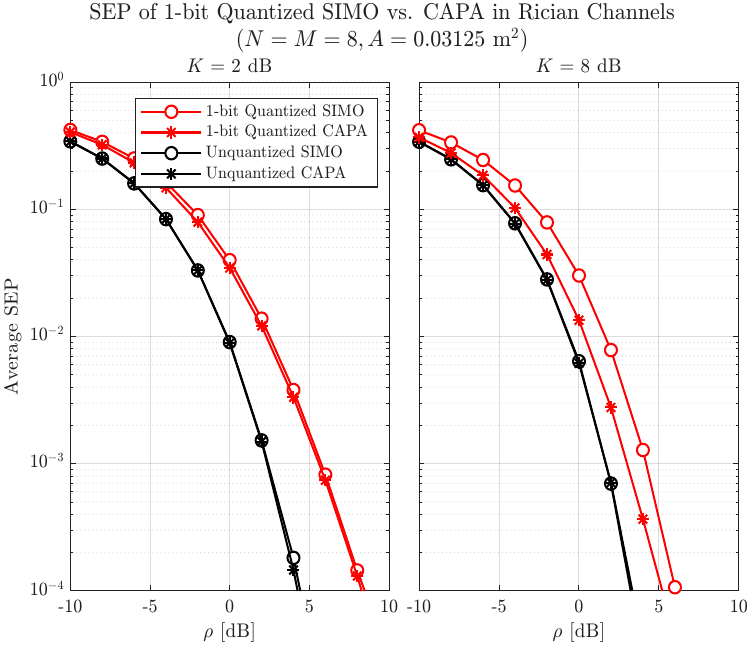}
        \vspace{-4ex}
        \caption{Rician fading.}
        \label{fig:1-bit_Rice8dB}
    \end{subfigure}
    \\
    \begin{subfigure}{0.95\columnwidth}
        \centering
        \includegraphics[width=\columnwidth]{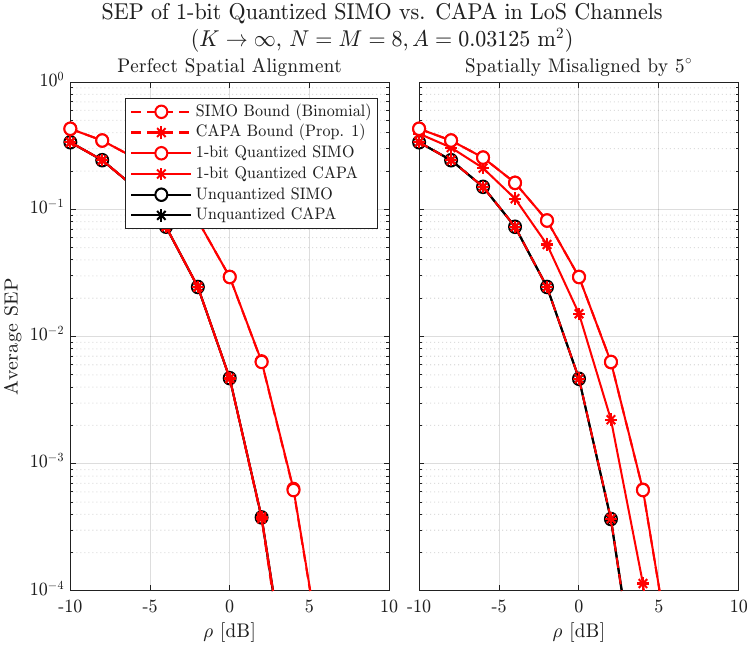}
        \vspace{-4ex}
        \caption{LoS channels.}
        \label{fig:1-bit_LoS_A1}
    \end{subfigure}
    \vspace{-1.5ex}
    \caption{Comparison of \ac{SEP} for \ac{CAPA} and \ac{SIMO} systems under various channel conditions.}
    \label{fig:1-bit_all}
\end{figure}

Most importantly, the 1-bit \ac{CAPA} curve tracks the unquantized limit exactly which confirms that the $2/\pi$ ADC penalty vanishes identically when the \ac{LoS} energy is confined to a single phase-aligned mode. 
The 1-bit \ac{SIMO} curve lags the unquantized reference by $8$--$10$~dB at $\mathrm{SEP} = 10^{-3}$, directly visualizing the ``twice the antennas'' penalty that \ac{CAPA} circumvents~\cite{ravinath2026sepjournal}.

The second subfigure in Fig.~\ref{fig:1-bit_LoS_A1} repeats the simulation with deliberate spatial misalignment. 
The \ac{SIMO} system is largely unaffected, since digital \ac{MRC} adapts to any \ac{LoS} phase profile. 
The 1-bit \ac{CAPA} performance degrades visibly: \ac{LoS} energy leaking into non-dominant modes contributes incoherently after quantization, partially restoring the \ac{ADC} penalty. 
The \ac{CAPA} curve falls between the \ac{AWGN} bound and the \ac{SIMO} curve, confirming that the conditions of Proposition~\ref{prop:Exact_LoS_Bounds} -- perfect spatial and phase alignment -- are necessary to achieve zero quantization loss.

\section*{Acknowledgment}

The authors would like to thank Prof. Christoph Studer for helpful discussions.



\bibliographystyle{IEEEtran}
\bibliography{references}

@STRING{IEEE_J_IT         = "{IEEE} Trans. Inf. Theory"}

@STRING{IEEE_J_WCOM       = "{IEEE} Trans. Wireless Commun."}

@STRING{IEEE_J_VT         = "{IEEE} Trans. Veh. Technol."}

@STRING{IEEE_M_COM        = "{IEEE} Commun. Mag."}

@STRING{IEEE_SPM        = "{IEEE} Signal Process. Mag."}

@ARTICLE{PeppasTVT2018,
  author={Peppas, Kostas P. and others},
  journal=IEEE_J_VT, 
  title={Approximations to the Distribution of the Sum of Generalized Normal RVs Using the Moments Matching Method and its Applications in Performance Analysis of Equal Gain Diversity Receivers}, 
  year={2018},
  volume={67},
  number={8},
  pages={7230-7241},
  doi={10.1109/TVT.2018.2830316}}

@ARTICLE{MaTIT2025,
  author={Ma, Yun and Wu, Yihong and Yang, Pengkun},
  journal=IEEE_J_IT, 
  title={On the Best Approximation by Finite Gaussian Mixtures}, 
  year={2025},
  volume={71},
  number={7},
  pages={5469-5492},
  doi={10.1109/TIT.2025.3558841}}

@ARTICLE{JacobssonTWC2017,
  author={Jacobsson, Sven and Durisi, Giuseppe and Coldrey, Mikael and Gustavsson, Ulf and Studer, Christoph},
  journal={IEEE Transactions on Wireless Communications}, 
  title={Throughput Analysis of Massive MIMO Uplink With Low-Resolution ADCs}, 
  year={2017},
  volume={16},
  number={6},
  pages={4038-4051},
  doi={10.1109/TWC.2017.2691318}}

@article{ravinath2026sepjournal,
  title={SEP Analysis of a Low-Resolution SIMO System with M-PSK over Fading Channels},
  author={Ravinath, Amila and Ding, Minhua and Gouda, Bikshapathi and Atzeni, Italo and T{\"o}lli, Antti},
  journal={arXiv preprint arXiv:2601.03387},
  year={2026}
}

@ARTICLE{PizzoJSAC2020,
  author={Pizzo, Andrea and Marzetta, Thomas L. and Sanguinetti, Luca},
  journal={IEEE Journal on Selected Areas in Communications}, 
  title={Spatially-Stationary Model for Holographic MIMO Small-Scale Fading}, 
  year={2020},
  volume={38},
  number={9},
  pages={1964-1979},
  doi={10.1109/JSAC.2020.3000877}}

@ARTICLE{AnnamalaiTCOM2000,
  author={Annamalai, A. and Tellambura, C. and Bhargava, V.K.},
  journal={IEEE Transactions on Communications}, 
  title={Equal-gain diversity receiver performance in wireless channels}, 
  year={2000},
  volume={48},
  number={10},
  pages={1732-1745},
  doi={10.1109/26.871398}}

@ARTICLE{HanTWC2023,
  author={Han, Zixiang and Shen, Shanpu and Zhang, Yujie and Tang, Shiwen and Chiu, Chi-Yuk and Murch, Ross},
  journal={IEEE Transactions on Wireless Communications}, 
  title={Using Loaded N-Port Structures to Achieve the Continuous-Space Electromagnetic Channel Capacity Bound}, 
  year={2023},
  volume={22},
  number={11},
  pages={7592-7605},
  doi={10.1109/TWC.2023.3253576}}

@ARTICLE{ZhaoTWC2025,
  author={Zhao, Boqun and Ouyang, Chongjun and Zhang, Xingqi and Shin, Hyundong and Liu, Yuanwei},
  journal={IEEE Transactions on Wireless Communications}, 
  title={Downlink and Uplink ISAC in Continuous-Aperture Array (CAPA) Systems}, 
  year={2025},
  volume={},
  number={},
  pages={1-1},
  doi={10.1109/TWC.2025.3605158}}

@ARTICLE{MoTSP2015,
  author={Mo, Jianhua and Heath, Robert W.},
  journal={IEEE Transactions on Signal Processing}, 
  title={Capacity Analysis of One-Bit Quantized MIMO Systems With Transmitter Channel State Information}, 
  year={2015},
  volume={63},
  number={20},
  pages={5498-5512},
  doi={10.1109/TSP.2015.2455527}}

@ARTICLE{OuyangTWC2025,
  author={Ouyang, Chongjun and others},
  journal={IEEE Transactions on Wireless Communications}, 
  title={Linear Receive Beamforming for CAPA Systems}, 
  year={2025},
  volume={},
  number={},
  pages={1-1},
  doi={10.1109/TWC.2025.3588626}}

@article{9906802,
	title        = {Wavenumber-Division Multiplexing in Line-of-Sight Holographic {MIMO} Communications},
	author       = {Sanguinetti, Luca and others},
	year         = 2023,
	month        = apr,
	journal      = IEEE_J_WCOM,
	volume       = 22,
	number       = 4,
	pages        = {2186--2201}
}

@article{dang2020should,
  title={What should 6{G} be?},
  author={Dang, Shuping and Amin, Osama and Shihada, Basem and Alouini, Mohamed-Slim},
  journal={Nature Electron.},
  volume={3},
  number={1},
  pages={20--29},
  year={2020},
  month = jan,
}

@article{larsson2014massive,
  title={Massive {MIMO} for next generation wireless systems},
  author={Larsson, Erik G and Edfors, Ove and Tufvesson, Fredrik and Marzetta, Thomas L},
  journal=IEEE_M_COM,
  volume={52},
  number={2},
  pages={186--195},
  year={2014}
}

@ARTICLE{10144733,
  author={Björnson, Emil and Eldar, Yonina C. and Larsson, Erik G. and Lozano, Angel and Poor, H. Vincent},
  journal=IEEE_SPM, 
  title={Twenty-Five Years of Signal Processing Advances for Multiantenna Communications: From theory to mainstream technology}, 
  year={2023},
  month=jun,
  volume={40},
  number={4},
  pages={107-117}
}

\end{document}